\documentclass{llncs}
\usepackage{cite}
\usepackage{amsmath,amssymb,amsfonts}
\usepackage{graphicx}
\usepackage[ruled,linesnumbered,lined,boxed,commentsnumbered]{algorithm2e}
\usepackage{booktabs}
\usepackage{tikz}
\usepackage{wasysym}
\usetikzlibrary{positioning}
\pgfdeclarelayer{background}
\pgfsetlayers{background,main}
\newcommand{\xv}{{\boldsymbol {x}}}
\newcommand{\vv}{{\boldsymbol {v}}}
\newcommand{\va}{{\boldsymbol {a}}}
\newcommand{\vs}{{\boldsymbol {s}}}

\newcommand{\VM}{{\it VM}}
\newcommand{\CV}{{\it CV}}
\newcommand{\UB}{{\it UB}}
\newcommand{\M}{{\mathcal{M}}}
\newcommand{\PD}{{\it PD}}

\DeclareMathOperator*{\argmin}{\arg\min}
\DeclareMathOperator*{\atan}{atan}

 \SetKwFunction{ampc}{AMPC} 
 \SetKwFunction{dampc}{DAMPC} 
 \SetKwFunction{localampc}{LAMPC}
 \SetKwFunction{distributedampc}{DAMPC}
 
 \SetKwFunction{paralleldampc}{DistributedAMPC} 
   
\usepackage{todonotes}

\usepackage{makeidx}  
\begin{document}
\frontmatter          
\pagestyle{headings}  

%
\mainmatter 

\title{Adaptive Neighborhood Resizing for \\Stochastic Reachability in Multi-Agent Systems}

\titlerunning{Adaptive Neighborhood Resizing}

\author{Anna~Lukina\inst{1} \and Ashish~Tiwari\inst{3} \and Scott.~A.~Smolka\inst{2} \and Radu~Grosu\inst{1,2}}

\authorrunning{Lukina, Tiwati, Smolka, and Grosu}

\institute{Cyber-Physical Systems Group, Technische Universit\"at Wien, Vienna, Austria\\
\and Department of Computer Science, Stony Brook University, New York, USA\\
\and Microsoft, USA\\
\email{anna.lukina@tuwien.ac.at}}

\maketitle

\begin{abstract}
We present $\distributedampc$, a distributed, adaptive-horizon and a\-dap\-tive-neighborhood algorithm for solving the stochastic reachability problem in multi-agent systems, in particular flocking modeled as a Markov decision process. 
At each time step, every agent calls a centralized, adaptive-horizon model-predictive control ($\ampc$) algorithm~\cite{attackingV} to obtain an optimal solution for its local neighborhood. Second, the agents derive the flock-wide optimal solution through a sequence of consensus rounds. Third, the neighborhood is adaptively resized using a flock-wide, cost-based Lyapunov function $V$. This way $\dampc$ improves efficiency without compromising convergence. 
We evaluate $\distributedampc$'s performance using statistical model checking. Our results demonstrate that, compared to $\ampc$, $\distributedampc$ achieves considerable speed-up (two-fold in some cases) with only a slightly lower rate of convergence. 
The smaller average neighborhood size and lookahead horizon demonstrate the benefits of the $\distributedampc$ approach for stochastic reachability problems involving any controllable multi-agent system that possesses a cost function.
\end{abstract}

\section{Introduction}
\label{sec:intro}
V-formation in a flock of birds is a quintessential example of emergent behavior in a stochastic multi-agent system. V-formation brings numerous benefits to the flock. It is primarily known for being energy-efficient due to the upwash benefit a bird in the flock enjoys from its frontal neighbor. In addition, it offers each bird a clear frontal view, unobstructed by any flockmate.  
Moreover, its collective spatial flock mass can be intimidating to potential predators. It is therefore not surprising that interest in V-formation is on the rise in the aircraft industry~\cite{bloomberg}.

Recent work on V-formation has shown that the problem can be viewed as one of optimal control, model-predictive control (MPC) in particular. In~\cite{attackingV}, we introduced adaptive-horizon MPC ($\ampc$), a highly effective control algorithm for multi-agent cyber-physical systems (CPS) modeled as a Markov decision process (MDP). Traditional MPC uses a fixed prediction horizon, i.e. number of steps to compute ahead, to determine the optimal, cost-minimizing control action. The downside of the fixed look-ahead is that the algorithm may get stuck in a local minimum.  For a controllable MDP, $\ampc$ chooses its prediction horizon dynamically, extending it out into the future until the cost function (shown in blue in Fig.~\ref{fig:example_local}) decreases sufficiently.  This implicitly endows $\ampc$ with a Lyapunov function (shown in red in Fig.~\ref{fig:example_local}), providing statistical guarantees of convergence to a goal state such as V-formation, even in the presence of adversarial agents. It should be noted that $\ampc$ works in a centralized manner, with global knowledge of the state of the flock at its disposal.

\begin{figure}[t]
\centering
\includegraphics[width=.5\textwidth]{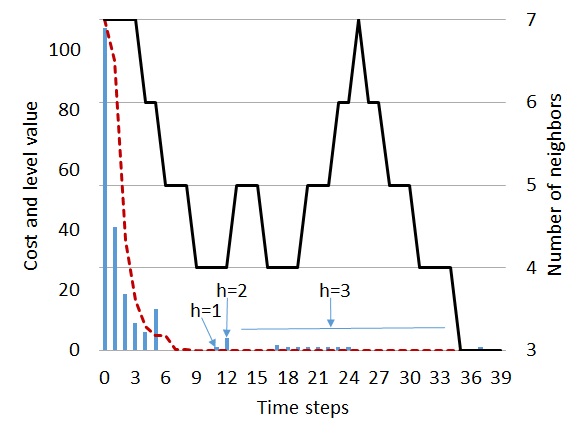}\includegraphics[width=.5\textwidth]{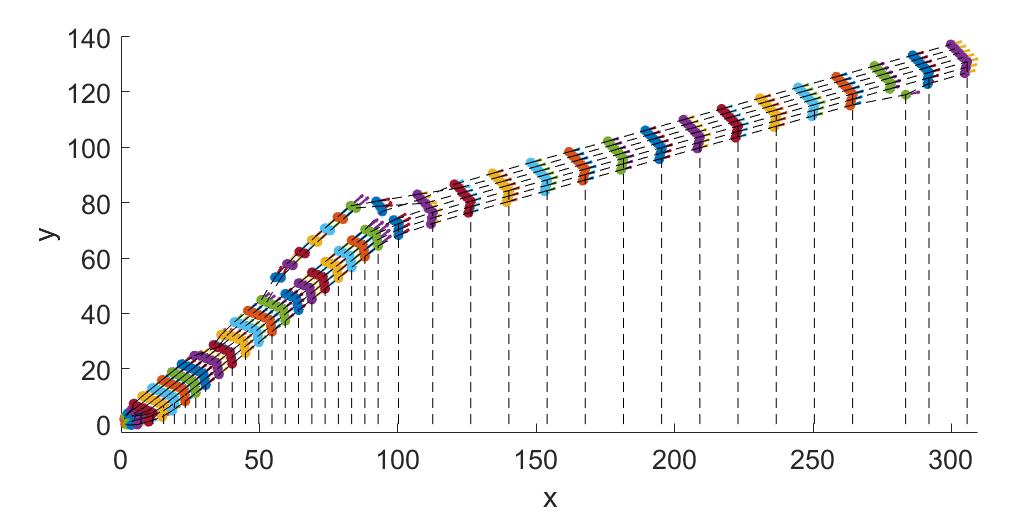}
\vspace*{-4ex}
\caption{Left: Blue bars are the values of the cost function in every time step. Red dashed line is the cost-based Lyapunov function used for horizon and neighborhood adaptation. Black solid line is neighborhood resizing for the next step given the current cost. Right: Step-by-step evolution of the flock of seven birds bringing two separate formations together. Each color-slice is a configuration of the birds at a particular time step.}
\label{fig:example_local}
\vspace*{-4ex}
\end{figure}

This paper introduces $\dampc$, a distributed version of $\ampc$ that extends it along several dimensions. First, at every time step, $\dampc$ runs a \emph{distributed consensus algorithm} to determine the optimal action (acceleration) for every agent in the flock. In particular, each agent $i$ starts by computing the optimal actions for its local subflock. The subflocks then communicate in a sequence of consensus rounds to determine the optimal actions for the entire flock.
Secondly, $\dampc$ features \emph{adaptive neighborhood resizing} (black line in Fig.~\ref{fig:example_local}) in an effort to further improve the algorithm's efficiency. In a similar way as for the prediction horizon in $\ampc$, neighborhood resizing utilizes the implicit Lyapunov function to guarantee eventual convergence to a minimum neighborhood size. $\distributedampc$ thus treats the neighborhood size as another controllable variable that can be dynamically adjusted for efficiency purposes. This leads to reduced communication and computation compared to the centralized solution, without sacrificing statistical guarantees of convergence like those offered by its centralized counterpart $\ampc$.

The proof of statistical global convergence is intricate.  For example, consider the scenario shown in Fig.~\ref{fig:example_local}. $\distributedampc$ is decreasing the neighborhood size $k$ for all agents, as the system-wide cost function $J$ follows a decreasing trajectory. Suddenly and without warning, the flock begins to split into two, undoubtedly owing to an unsuitably low value of $k$, leading to an abrupt upward turn in $J$. $\distributedampc$ reacts accordingly and promptly, increasing its prediction horizon first and then $k$, until system stability is restored. The ability for $\distributedampc$ to do this is guaranteed, for in the worst case $k$ will be increased to $B$, the total number of birds in the flock. It can then again attempt to monotonically decrease $k$, but this time starting from a lower value of $J$, until V-formation is reached.

A smoother convergence scenario is shown in Fig.~\ref{fig:example}. In this case, the  efficiency gains of adaptive neighborhood resizing are more evident, as the cost function $J$ follows an almost purely monotonically decreasing trajectory. A formal proof of global convergence of $\distributedampc$ with high probability is given in the body of the paper, and represents one of the paper's main results.

Apart from the novel adaptive-horizon adaptive-neighborhood distributed algorithm to synthesize a controller, and its verification using statistical model checking, we believe the work here is significant in a deeper way. The problem of synthesizing a sequence of control actions to drive a system to a desired state can be also viewed as a falsification problem, where one tries to find values for (adversarial) inputs that steer the system to a bad state. 

These problems can be cast as constraint satisfaction problems, or as optimization problems. As in case of V-formation, one has to deal with non-convexity, and popular techniques, such as convex optimization, will not work. Our approach can be seen as a tool for solving such highly nonlinear optimization problems that encode systems with notions of time steps and spatially distributed agents. Our work demonstrates that a solution can be found efficiently by adaptively varying the time horizon and the spatial neighborhood. A main benefit of the adaptive scheme, apart from efficiency, is that it gives a path towards completeness. By allowing adaptation to consider longer time horizons, and larger neighborhoods (possibly the entire flock), one can provide convergence guarantees that would be otherwise impossible (say, in a fixed-horizon MPC).

The rest of the paper is organized as follows. Section~\ref{sec:related} discusses related work. 
Section~\ref{sec:v-form} describes the flocking model and the cost function. Section~\ref{sec:reach} defines the stochastic reachability problem. Section~\ref{sec:dist_flock} introduces $\distributedampc$, our main contribution. Section~\ref{sec:convergence} provides proofs of our theoretical results. Section~\ref{sec:results} presents our statistical evaluation of the algorithm. Section~\ref{sec:conclusion} offers concluding remarks and indicates directions for future research.

\section{Related Work} 
\label{sec:related}


In~\cite{lukina_tacas17}, the ARES algorithm generates plans incrementally, segment by segment, using adaptive horizons (AH) to find the next best step towards the global optimum. In particular, ARES calls a collection of particle swarm optimizers (\texttt{PSO}), each with its own AH. \texttt{PSO}s with the best results are cloned, while the others are restarted from the current level on the way to the goal (similar to importance splitting~\cite{kahn1951estimation}). When the V-formation is achieved, the complete plan is composed of the best segments. The AHs are chosen such that the best \texttt{PSO}s can succeed to decrease the objective cost by at least a pre-defined value, implicitly defining a Lyapunov function that guarantees global convergence.

The presence of an adversary able to disturb the state of the system at every time step (e.g., remove a bird) is investigated in~\cite{attackingV}. In this case, planning is not sufficient. Instead, a controller ($\ampc$) finding the best accelerations at every step has to be designed. This calls only one \texttt{PSO} with a given horizon and uses the first flock-wide accelerations (from the sequence returned by \texttt{PSO}) as the next actions. Under the assumption that the flock is controllable, even if the returned accelerations are not the optimal (in terms of the entire run), $\ampc$ can correct for this in the future. AH is again picked based on the distance of the system to the goal states. This results in a global $\ampc$ that, unlike classical MPC (which uses a fixed horizon), is guaranteed to converge.

In~\cite{ZhanL13}, the problem of taking an arbitrary initial configuration of $n$ agents to a final configuration, where every pair of stationary ``neighbors'' is a fixed distance $d$ apart, is considered. They present centralized and distributed algorithms for this problem, both of which use MPC to determine the next action. The problem in~\cite{ZhanL13} is related to our work. However, we consider nonconvex and nonlinear cost functions, which require overcoming local minima to ensure convergence. In contrast,~\cite{ZhanL13} deals with convex functions, which do not suffer from problems introduced by the presence of multiple local minima. Furthermore, in the distributed control procedure of~\cite{ZhanL13}, each agent publishes the control value it locally computed, which is then used by other agents to calculate their own. A quadratic number of such steps is performed before each agent fixes its control input for the next time step. In our work, we limit this number to linear.

Other related work, including~\cite{Fowler02,Dandrea03,Ye17}, focuses on distributed controllers for flight formation that operate in an environment where the multi-agent system is already in the desired formation and the distributed controller's objective is to maintain formation in the presence of disturbances. A distinguishing feature of these approaches is the particular formation they are seeking to maintain, including a half-vee~\cite{Fowler02}, a ring and a torus~\cite{Dandrea03}, and a leader-follower formation~\cite{Ye17}. These works are specialized for capturing the dynamics of moving-wing aircraft. In contrast, we use $\dampc$ with dynamic neighborhood resizing to bring a flock from a random initial configuration to a stable V-formation.

Although $\distributedampc$ uses global consensus, our main focus is on adaptive neighborhood resizing and global convergence, and not on fault tolerance~\cite{su2016fault,dolev1986reaching}.

\section{Background on V-Formation}
\label{sec:v-form}

\paragraph{\rm\textbf{Dynamical model.}}In the flocking model used, the state of each bird is given by four variables: a 2-dimensional vector $\xv$ denoting the position of the bird in 2D continuous space, and a 2-dimensional vector $\vv$ denoting the velocity of the bird. We use $\vs=\{\xv_i, \vv_i\}_{i=1}^B$ to denote a state of a flock with $B$ birds. The \emph{control actions} of each bird are 2-dimensional accelerations $\va$.

Let $\xv_i(t)$, $\vv_i(t)$, and $\va_i(t)$ 
denote the position, velocity, and acceleration, 
of $i$-th bird at time $t$, $i\,{\in}\,\{1,\ldots,B\}$, respectively. Given an initial configuration $\xv_i(0) = \xv_i^0, \vv_i(0)=\vv_i^0$ inside a bounding box of a given size, the discrete-time behavior of bird $i$ is given by Eq.~\ref{eq:trans}:
\begin{eqnarray}
\label{eq:trans}
 \vv_i(t+1) &=& \vv_i(t) + \va_i(t),\label{eq:v} \qquad  \nonumber\\
 \xv_i(t + 1) &=& \xv_i(t) + \vv_i(t).\label{eq:x} 
\end{eqnarray}

In the simulations described in Section~\ref{sec:results}, the time interval between two successive updates of the positions and velocities of all birds in the flock equals one. The initial state is generated uniformly at random inside a bounding box. The accelerations are the output of the particle swarm optimization (\texttt{PSO}) algorithm~\cite{Kennedy95particleswarm}, which samples particles uniformly at random subject to the following constraints on the maximum velocities and accelerations: $\forall$ $i\,{\in}\,\{1,\ldots,B\}$ $||\vv_i(t)||\,{\leqslant}\,\vv_{max}$, $||\va_i(t)||\,{\leqslant}\,\rho||\vv_i(t)||$, where $\vv_{max}$ is a constant and $\rho\,{\in}\,(0,1)$. 

Thus, $\va$ introduces uncertainty in our model. We chose this model for its simplicity, as the cost function described in Section \ref{sec:fitness}, is nonlinear, nonconvex, nondifferentiable, and therefore sufficient enough to make reachability analysis extremely challenging. We propose an approximate algorithm to tackle reachability. In principle, more sophisticated models of flocking dynamics can be considered, but we leave those for future work, and focus on the simplest one.

The problem of bringing a flock from an arbitrary configuration to a V-formation can be posed as a reachability question, where the goal is the set of states representing a V-formation. A key assumption is that the reachability goal can be specified as $J(\vs) \leqslant \varphi$, where $J$ is a cost function that assigns a nonnegative real value to each state $\vs$, and $\varphi$ is a small positive constant.

\paragraph{\rm\textbf{Cost Function.}}
\label{sec:fitness}
In order to define the cost function, we recall the definitions of the metrics determining the cost of a state from~\cite{lovethyneighbor16} (see Appendix for details).
\begin{itemize}
\item \emph{Clear View:} $\CV(\vs)$ is defined by accumulating the percentage of a cone with angle $\theta$, blocked by other birds. The minimum value is $\CV^*{=}\,0$ and attained in a perfect V-formation where all birds have an unobstructed view.

\item \emph{Velocity Matching:} $\VM(\vs)$ for flock state $\vs$ is defined as the difference between the velocity of a given bird and all other birds, summed up over all birds in the flock. The minimum value is $\VM^*{=}\,0$ and attained in a perfect V-formation where all birds have the same velocity.

\item \emph{Upwash Benefit:} the trailing upwash is generated near the wingtips of a bird, while downwash is in the center of a bird. An upwash measure is defined on the 2D space using a Gaussian-like model that peaks at the appropriate upwash and downwash regions. $\UB(\vs)$ for flock state $s$ is the sum of $\UB_i$ for $1\leqslant i \leqslant B$. The upwash benefit $\UB(\vs)$ in V-formation is $\UB^*\,{=}\,1$, as all birds, except for the leader, have minimum upwash-benefit metric ($\UB_i=0$), while the leader has an upwash-benefit metric of $1$ ($\UB_i=1$).
\end{itemize}

Given the above metrics, the overall objective function $J$ is defined as a sum-of-squares of $\VM$, $\CV$, and $\UB$, as follows:
\begin{align}
\label{eq:cost}
J(\vs) = (\CV(\vs)-\CV^*)^2 + (\VM(\vs)-\VM^*)^2 +(\UB(\vs)-\UB^*)^2.
\end{align}
A state $\vs^{*}$ is considered to be a V-formation if $J(\vs^{*})\,{\leqslant}\,\varphi$, for a small positive $\varphi$. 

\section{The Stochastic Reachability Problem}
\label{sec:reach}

Given the stochasticity introduced by \texttt{PSO}, the V-formation problem can be formulated in terms of a reachability problem for a Markov Chain, induced by the composition of a Markov decision process (MDP) and a controller.  

\begin{definition}A \textbf{Markov decision process} (MDP) $\M=(S,A,T,J,I)$ is a 5-tuple consisting of a set of states $S$, a set of actions $A$, a transition function $T: S\,{\times}\,A\,{\times}\,S\,{\mapsto}\,[0,1]$, where $T(\vs,a,\vs^\prime)$ is the probability of transitioning from state $\vs$ to state $\vs'$ under action $\va$, a cost function $J:S\,{\mapsto}\,\mathbb{R}$, where $J(\vs)$ is the cost associated with state $\vs$, and an initial state distribution $I: S\,{\mapsto}\,[0,1]$.
\end{definition}

The \emph{MDP $\M$ modeling a flock} of $B$ birds is defined as follows. The set of states $S$ is $S = \mathbb{R}^{4B}$, as each bird has a $2$D position and a $2$D velocity vector, and the flock contains $B$ birds. The set of actions $A$ is $A = \mathbb{R}^{2B}$, as each bird takes a $2$D acceleration action and there are $B$ birds. The cost function  $J$ is defined by Eq.~\ref{eq:cost}. The transition function $T$ is defined by Eq.~\ref{eq:v}. As the acceleration vector $\va_i(t)$ for bird $i$ at time $t$ is a random variable, the state vector $(\xv_i(t+1)$, $\vv_i(t+1))$ is also a random variable. The initial state distribution $I$ is a uniform distribution from a region of state space where all birds have positions and velocities in a range defined by fixed lower and upper bounds.

Before we can define traces, or executions, of $\M$, we need to fix a controller, or strategy, that determines which action from $A$ to use at any given state of the system. We focus on randomized strategies. A \emph{randomized strategy} $\sigma$ over $\M$ is a function of the form $\sigma: S\,{\mapsto}\,\PD(A)$, where $\PD(A)$ is the set of probability distributions over $A$. That is, $\sigma$ takes a state $\vs$ and returns an action consistent with the probability distribution $\sigma(\vs)$. Once we fix a strategy for an MDP, we obtain a Markov chain. We refer to the underlying Markov chain induced by $\sigma$ over $\M$ as $\M_{\sigma}$. We use the terms strategy and controller interchangeably.

In the bird-flocking problem, a controller would be a function that determines the accelerations for all the birds given their current positions and velocities. Once we fix a controller, we can iteratively use it to (probabilistically) select a sequence of flock accelerations. The goal is to generate a sequence of actions that takes an MDP from an initial state $\vs$ to a state $\vs^{*}$ with $J(\vs^{*})\,{\leqslant}\,\varphi$. 


\begin{definition} Let $\M\,{=}\,(S,A,T,J,I)$ be an MDP, and let $G \subseteq S$ be the set of goal states $G\,{=}\,\{\vs | J(\vs)\,{\leqslant}\,\varphi\}$ of $\M$. Our \textbf{stochastic reachability problem} is to design a controller $\sigma: S\,{\mapsto}\,\PD(A)$ for $\M$ such that for a given $\delta$ probability of the underlying Markov chain $\M_\sigma$ to reach a state in $G$ in $m$ steps, for a given $m$, starting from an initial state, is at least $1-\delta$.
\end{definition}

We approach the stochastic reachability problem by designing a controller and quantifying its probability of success in reaching the goal states.
In~\cite{lukina_tacas17}, a stochastic reachability problem was solved by appropriately designing centralized controllers $\sigma$. In this paper, we design a distributed procedure with an adaptive horizon and adaptive neighborhood resizing and evaluate its performance.

\section{Adaptive-Neighborhood Distributed Control}
\label{sec:dist_flock}
In contrast to \cite{lukina_tacas17,attackingV}, we consider a distributed setting with the following assumptions about the system model. 
\begin{enumerate}\itemsep=0em
\item Each bird is equipped with the means for communication. The communication radius of each bird $i$ changes its size adaptively.
The measure of the radius is the number of birds covered and we refer to it as the bird's local neighborhood $N_i$, including the bird itself.
\item All birds use the same algorithm to satisfy their local reachability goals, i.e. to bring the local cost $J(\vs_{N_i})$, $i\,{\in}\,\{1,\ldots,B\}$, below the given threshold $\varphi$.
\item 
The birds move in continuous space and change accelerations synchronously at discrete time points.
\item After executing its local algorithms, each bird broadcasts the obtained solutions to its neighbors. This way every bird receives solution proposals, which differ due to the fact that each bird has its own local neighborhood. To find consensus, each bird takes as its best action the one with the minimal cost among the received proposals. The solutions for the birds in the considered neighborhood are then fixed. The consensus rounds repeat until all birds in the flock have fixed solutions.
\item Every time step the value of the cost function $J(\vs)$ is obtained globally for all birds in the flock and checked for improvement. The neighborhood for each bird is then resized based on this global check.
%
%
\vspace*{1mm}\item The upwash modeled in Section \ref{sec:fitness} maintains connectivity of the flock along the computations, while our algorithm manages collision avoidance.
\end{enumerate}

%
%

The main result of this paper is a \emph{distributed adaptive-neighborhood and adaptive-horizon model-predictive control algorithm} we call $\distributedampc$. At each time step, each bird runs \ampc to determine the best acceleration for itself and its neighbors (while ignoring the birds outside its neighborhood). 
The birds then exchange the computed accelerations with their neighbors, and the whole flock 
arrives at a consensus that assigns each bird to a unique (fixed) acceleration value.  
Before reaching consensus, it may be the case that some of $i$'s neighbors already have fixed solutions (accelerations) -- these accelerations are not updated when $i$ runs \ampc. A key idea of our algorithm is to adaptively resize the extent of a bird's neighborhood. 



\subsection{The Distributed AMPC Algorithm}
\label{sec:main_algo}
%
\begin{algorithm}[ht!]
\small
	\SetKwFunction{localampc}{LocalAMPC}  
    \SetKwFunction{cost}{cost}
    \SetKwFunction{J}{J}
    \SetKwFunction{fixed}{Fixed}
    \SetKwFunction{fix}{Fix}
    \SetKwFunction{neigh}{Neighbors}
    \SetKwFunction{neighsize}{NeighSize}
    
	\SetKwInOut{Input}{Input}
	\SetKwInOut{Output}{Output}
    \SetKwFor{ParFor}{for}{do in parallel}{endfor} 
    \DontPrintSemicolon
	\Input{$\M\,{=}\,(S,A,T,J,I),\varphi,{h}_{\mathit{max}},m,B, \beta$}
	\Output{$\vs_0$, $\va^m\,{=}\,[\va(t)]_{1\leqslant t\leqslant\,m}$} 
	\BlankLine
	$\vs_0\leftarrow\mathrm{sample}(I)$; $\vs\leftarrow\vs_0$; 
    $\ell_0\leftarrow J(\vs)$; $[\widehat{J}_i]_{i\in{B}}\leftarrow\inf$;  $t\leftarrow 1$; $k\leftarrow B$; \tcp*{Init}
	\BlankLine
	\While{($\ell_{t{-}1} > \varphi)$ $\land$ $(t < m)$}
	{
      $\overline{\va}(t)\leftarrow\textit{nfy}$; \tcp*{Initially no bird has a fixed solution}
      \While (\tcp*[f]{while not all birds have a fixed solution}) {$\lnot \fixed(\overline{\va}(t)) $} 
      {
        $R\leftarrow\{j \,|\, 1 \leqslant j \leqslant B \land \lnot \fixed(\overline{\va}_{j}(t))\}$;\;
        \tcp{Birds without a fixed solution run \localampc}
     	\ParFor{ $i \in R$ } 
        {
          $N_i \leftarrow \neigh(i, k)$; \tcp{Find $k$ nearest neighbors of $i$}
          $\Delta_i \leftarrow J(\vs_{N_i})/(m{-}t)$;\;
          $(\widehat{\vs}_{N_i},\tilde{\vs}_{N_i},\va_{N_i}^{h_i},\widehat{J}_i)\leftarrow\localampc(\M,\vs_{N_i},\overline{\va}_{N_i},\Delta_i,h_{max},\beta)$;
        }  
        $i^* \leftarrow \argmin_{i\in R}{\widehat{J}_i}$; \tcp{Find the bird with the best solution}    
        \ForAll{$b\in\neigh(i^*,k)$  \tcp{Fix $i^*$'s neighbors solution}}  
        {
          $\overline{\va}_b(t)\leftarrow\va_{N_{i^*}}^{h_i}[b]$; \tcp*{$\va^{h_i}_b(t)$ is the solution for bird $b$}
        }
      }
      $\va(t)\leftarrow\mathrm{first}(\overline{\va}(t))$; $\vs\leftarrow\widetilde{\vs}$;
      \tcp*{First action and next state}
      \If {$\ell_{t{-}1} - J(\widehat{\vs}) >\Delta$} {
        $\ell_t\leftarrow J(\widehat{\vs})$; $t\leftarrow t{+}1$; \tcp*{Proceed to the next level}
      }
      $k \leftarrow \neighsize(J(\widehat{\vs}), k)$; \tcp*{Adjust the neighborhood size}
	}
	\caption{DAMPC}
	\label{alg:distributed}
\end{algorithm}
\setlength{\floatsep}{1cm}

$\distributedampc$ (see Alg.~\ref{alg:distributed}) takes as input an MDP $\M$, a threshold $\varphi$ defining the goal states $G$, the maximum horizon length $h_{max}$, the maximum number of time steps $m$, the number of birds $B$, and a scaling factor $\beta$. It outputs a state $\vs_0$ in $I$, and a sequence of actions $\va^m$ taking $\M$ from $\vs_0$ to a state in $G$.

The initialization step (Line 1) picks an initial state $\vs_0$ from $I$, fixes the initial level $\ell_0$ as the cost of $\vs_0$, sets an arrays of costs to infinite values, sets the initial time, and sets the number of birds to process. 

The outer while loop (Lines 2-23) is active as long as $\M$ has not reached $G$ and time has not expired. In each time step, $\distributedampc$ first sets the sequences of accelerations $\overline{\va}(t)$ to ``not fixed yet'' (\texttt{nfy}), and then iterates (Lines 4-17) until all birds fix their accelerations through global consensus. This happens as follows. First, all birds determine their neighborhood (\textit{subflock}) $N_i$ and the cost decrement $\Delta_i$ that will bring them to the next level (Lines 8-9). Second, they call \localampc (see Section~\ref{sec:lampc}), which returns (Line 10): a sequence of actions $\va_{N_i}^{h_i}$ of length $h_i$ for the subflock, which should decrease the subflock cost by $\Delta_i$, the state $\widetilde{\vs}_{N_i}$ of the subflock after executing the first action in $\va_{N_i}^{h_i}$, the state $\widehat{\vs}_{N_i}$ after executing the last action, and the cost $\widehat{J}_i$ in the last state. Third, they determine the subflock $N_{i^{*}}$ with lowest cost as a winner (Line 12)\footnote{This step requires global consensus, but more generally, the loop on Lines 4-17 requires birds to have global information at multiple places. This can be quite inefficient in practice. A more practical approach, given in~\cite{saber2003consensus}, is based on a dynamic local (among neighbors) consensus, for fixed neighborhood graphs. Since we are adapting and changing neighborhood sizes, the results are not directly applicable. Nevertheless, we can still use a similar truly distributed approach (in place of global consensus on Lines 4-17), but preferred to experiment with the easier to implement global version, since adapting neighborhoods size is the main focus of our paper, and our goal was to evaluate if neighborhood sizes really shrink, and remain small as the flock converges to a state $\vs$ in $G$, with $J(\vs)\leqslant\varphi$.} and fix the acceleration sequences of all birds in this subflock (Lines 13-16).

After all accelerations sequences are fixed, that is $\texttt{Fixed}(\overline{\va}(t))$ is true, the first accelerations in this sequence are selected for the output (Line 18). The next state $\vs$ is set to $\tilde{\vs}$, the state of the flock after executing $\va(t)$, and $\hat{\vs}$ has the state after executing last action in $\overline{\va}(t)$. If we found a path that eventually decreases the cost by $\Delta$, we reached the next level, and advance time (Lines 19-21). In that case, we optionally decrease the neighborhood, and increase it otherwise (Line 22).

\begin{figure}[t]
\centering

\def\lav{cyan!50}
\def\oran{yellow!50}

\tikzstyle{birds}=[draw,circle,black,fill=\lav,
                  text=black,minimum width=10pt]
\tikzstyle{fixedbirds}=[draw,rectangle,black,fill=black,
                  text=white,minimum width=10pt]
\tikzstyle{chosenbirds}=[draw,circle, black, fill=\oran,
                       text=black,minimum width=25pt]
\tikzstyle{legendcb}=[rectangle, draw, black, rounded corners,
                     thin,fill=\oran,
                     text=black, minimum width=2.5cm]
\tikzstyle{legendfb}=[rectangle, draw, black, rounded corners,
                     thin,fill=black,
                     text=white, minimum width=2.5cm]
\tikzstyle{legendb}=[rectangle, draw, black, rounded corners, thin,fill=\lav,
                     text= black, minimum width= 2.5cm]
\tikzstyle{legend_general}=[rectangle, rounded corners, thin,
                           black, fill= white, draw, text=black,
                           minimum width=2.5cm, minimum height=0.8cm]

  
  
\begin{tikzpicture}[auto, thick]
  \node[birds] (a1) at (0,0) {$1$};
  \node[chosenbirds] (a2) at (0.1,1) {$2$};
  \node[fixedbirds] (a3) at (0.5,2) {$3$};
  \node[fixedbirds] (a4) at (1.3,1.3) {$4$};
  \node[fixedbirds] (a6) at (1.2,0.1) {$6$};
  \node[fixedbirds] (a7) at (2.6,0) {$7$};
  \node[fixedbirds] (a5) at (1.8,1.4) {$5$};
  
  \path (a2) edge (a1);
  \path (a2) edge (a3);
  \path (a2) edge (a4);
  \path (a2) edge (a6);
  
  \draw (-0.8,-0.6) rectangle (3,2.5) {};
\end{tikzpicture}
%
%
%
\begin{tikzpicture}[auto, thick]
  \node[fixedbirds] (a1) at (-0.5,0.7) {$1$};
  \node[fixedbirds] (a2) at (-0.2,1.5) {$2$};
  \node[fixedbirds] (a3) at (0.4,2.1) {$3$};
  \node[fixedbirds] (a4) at (1,1.5) {$4$};
  \node[birds] (a6) at (1.9,0.5) {$6$};
  \node[birds] (a7) at (2.6,0) {$7$};
  \node[chosenbirds] (a5) at (1.8,1.4) {$5$};
  
  \path (a5) edge (a6);
  \path (a5.north) edge (a3);
  \path (a5) edge (a4);
  
  \draw (-0.8,-0.6) rectangle (3,2.5) {};
\end{tikzpicture}
\begin{tikzpicture}[auto, thick]
  \node[fixedbirds] (a1) at (-0.5,0.7) {$1$};
  \node[fixedbirds] (a2) at (-0.2,1.5) {$2$};
  \node[fixedbirds] (a3) at (0.4,2.1) {$3$};
  \node[fixedbirds] (a4) at (1,1.5) {$4$};
  \node[fixedbirds] (a6) at (1.9,0.5) {$6$};
  \node[chosenbirds] (a7) at (2.6,0) {$7$};
  \node[fixedbirds] (a5) at (1.8,1.4) {$5$};
  
  \path (a7.north) edge (a5);
  \path (a7.north) edge (a6);
  \path (a7.north) edge (a4);
  
  \draw (-0.8,-0.6) rectangle (3.1,2.5) {};
\end{tikzpicture}
\caption{Left: Last round of consensus for neighborhood size four where Bird~2 runs Local AMPC taking as an input for \texttt{PSO} fixed accelerations of Birds~3,~4, and~6 together with \texttt{nfy} value for Bird~1. Middle: Second consensus round during the next time step where the neighborhood size was reduced to three as a result of the decreasing cost at the previous time step. Right: Third consensus round during the same time step where Bird~7 is the only one whose acceleration has not been fixed yet and it simply has to compute the solution for its neighborhood given fixed accelerations of Birds~4,~5, and~6.}
\label{fig:compare1}
\end{figure}
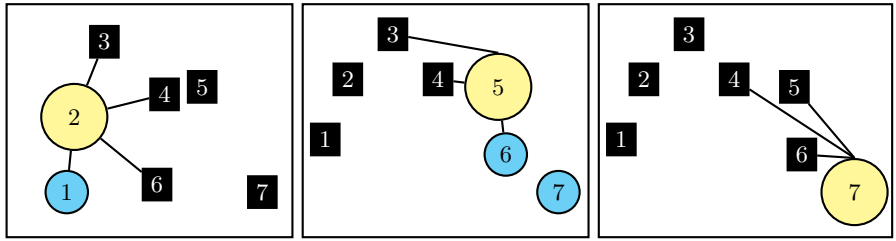
Fig.~\ref{fig:compare1} illustrates $\distributedampc$ for two consecutive consensus rounds after neighborhood resizing. Bigger yellow circles represent birds that are running $\localampc$. Smaller blue circles represent birds whose acceleration sequences are not completely fixed yet. Black squares mark birds with already fixed accelerations. Connecting lines are neighborhood relationship. 

Working with a real CPS flock requires careful consideration of energy consumption. Our algorithm accounts for this by using the smallest neighborhood necessary during next control input computations. Regarding deployment, we see the following approach. Alg.~\ref{alg:local_ampc} can be implemented as a local controller on each drone and communication will require broadcasting positions and output of the algorithm to other drones in the neighborhood through a shared memory. In this case, according to Alg.~\ref{alg:distributed}, a central agent will be needed to periodically compute the global cost and resize the neighborhood. Before deployment, we plan to use OpenUAV simulator~\cite{DBLP:conf/iccps/SchmittleLVDBRS18} to test \dampc on drone formation control scenarios described in~\cite{DBLP:conf/iccps/LukinaKSSDRBSG018}.

\subsection{The Local AMPC Algorithm}
\label{sec:lampc}
$\localampc$ is a modified version of the \texttt{AMPC} algorithm\cite{attackingV}, as shown in Alg.~\ref{alg:local_ampc}. 
Its input is an MDP $\M$, the current state $\vs_{N_i}$ of a subflock $N_i$, a vector of acceleration sequences $\va_{N_i}^{h_i}$, one sequence for each bird in the subflock, a cost decrement $\Delta$ to be achieved, a maximum horizon $h_{max}$ and a scaling factor $\beta$. In $\va_{N_i}^{h_i}$ some accelerations may not be fixed yet, that is, they have value \texttt{nfy}.

Its output is a vector of acceleration sequences $\overline{\va}^{*}$, one for each bird, that decreased the cost of the flock at most, the state $\vs_{N_i}$ of the subflock after executing the first action in the sequence, the state $\widehat{\vs}_{N_i}$ after executing all actions, and the cost $J_i$ actually achieved by the subflock in state $\widehat{\vs}_{N_i}$.
%
\begin{algorithm}[ht!]
\small
	\SetKwFunction{PSO}{PSO}
    \SetKwFunction{cost}{cost}
	\SetKwInOut{Input}{Input}
	\SetKwInOut{Output}{Output}
    \DontPrintSemicolon
	\Input{$\M\,{=}\,(S,A,T,J,I)$, $\vs_{N_i}$, $\va_{N_i}^{h_i}$, $\Delta$, $h_{max}$, $\beta$}
	\Output{$\vs_{N_i}$, $\widehat{\vs}_{N_i}$, $\overline{\va}^*$, $\widehat{J}_i$}
	\BlankLine
    ${p}\leftarrow{2}\cdot\beta\cdot{h_i}\cdot{B}$; $\widehat{J}_i\leftarrow$ Inf; \tcp{Initialization}
	\BlankLine
	\While{$(J(\vs_{N_i})-\widehat{J}_{i} < \Delta) \land (h_i \leqslant h_{max})$}
	{
       \tcp{Run PSO with local information $\vs_{N_i}$ and partial solution $\va_{N_i}^h$}  
       $(\vs_{N_i},\widehat{\vs}_{N_i},\overline{\va}^*)\leftarrow$\PSO{$\M,\vs_{N_i},\va_{N_i}^{h_i},p,h_i$};\\
	   $\widehat{J}_i\leftarrow{J}(\widehat{\vs}_{N_i})$;
       ${h_i}\leftarrow{h_i}+1$; 
       $p\leftarrow{2}\cdot\beta\cdot{h_i}\cdot{B}$; \tcp{increase horizon}
    }
	\caption{LocalAMPC}
	\label{alg:local_ampc}
\end{algorithm}
$\localampc$ first initializes (Line 1) the number of particles $p$ to be used by the particle swarm optimization algorithm (\texttt{PSO}), proportionally to the input horizon $h_i$ of the input accelerations $\va_{N_i}^{h_i}$, to the number of birds $B$, and the scaling factor $\beta$. It then tries to decrement the cost of the subflock by at least $\Delta$, as long as the maximum horizon $h_{max}$ is not reached (Lines 2-6). 

For this purpose it calls \texttt{PSO} (Line 4) with an increasingly longer horizon, and an increasingly larger number of particles. The idea is that the flock might have to first overcome a cost bump, before it gets to a state where the cost decreases by at least $\Delta$. \texttt{PSO} extends the input sequences of fixed actions to the desired horizon with new actions that are most successful in decreasing the cost of the flock, and it  computes from scratch the sequence of actions, for the \texttt{nfy} entries. The result is returned in $\overline{\va}^{*}$. \texttt{PSO} also returns the states $\vs_{N_i}$ and $\widehat{\vs}_{N_i}$ of the flock after applying the first and the last actions, respectively. Using this information, it computes the actual cost achieved by the flock.

\begin{lemma}[Local Convergence]
\label{lem:local_ampc}
Given $\M=(S,A,T,J,I)$, an MDP with cost function $\texttt{cost}$, and a nonempty set of target states $G\subset S$ with $G=\{\vs\,|\,J(\vs)\leqslant\varphi\}$. If the transition relation $T$ is controllable with actions in $A$ for every (local) subset of agents, then there exists a finite (maximum) horizon $h_{max}$ such that \localampc is able to find the best actions $\va_{N_i}^*$ that decreases the cost of a neighborhood of agents in the states $\vs_{N_i}$ by at least a given $\Delta$.
\end{lemma}
\begin{proof}
In the input to \localampc,
the accelerations of some birds in $N_i$ may be fixed (for some horizon). As a consequence, the MDP $\M$ may not be fully controllable within this horizon. Beyond this horizon, however, \texttt{PSO} is allowed to freely choose the accelerations, that is, the MDP $\M$ is fully controllable again. The result now follows from convergence of \texttt{AMPC} (Theorem~1 from~\cite{attackingV}).\qed
\end{proof}

\subsection{Dynamic Neighborhood Resizing}
\label{sec:resizing}

The key feature of $\distributedampc$ is that it \emph{adaptively resizes neighborhoods}. This is based on the following observation: \emph{as the agents are gradually converging towards a global optimal state, they can explore smaller neighborhoods} when computing actions that will improve upon the current configuration. 


Adaptation works on lookahead cost, which is the cost that is reachable in some future time. Line~20 of $\distributedampc$ is reached (and the level $t$ is incremented) whenever we are able to decrease this look-ahead cost. If level $t$ is incremented, neighborhood size $k\in [k_{min},k_{max}]$ is decremented, and incremented otherwise, as follows: $\neighsize(J, k) =$ 
\begin{align}
\begin{split}
\begin{cases}
   \min\left(\max\left(k-\left\lceil(1-\frac{J(s(t))}{k}) \right\rceil,k_{min}\right), k_{max}\right) & \text{if level $t$ was incremented} \\  
   \min\left(k+1, k_{max}\right) & \text{otherwise}.
  \end{cases}
\end{split}\label{eq:size}
\end{align}
In Fig.~\ref{fig:example} we depict a simulation-trace example, demonstrating how levels and neighborhood size are adapting to the current value of the cost function.

\begin{figure}[t]
\centering	
\includegraphics[width=.5\textwidth]{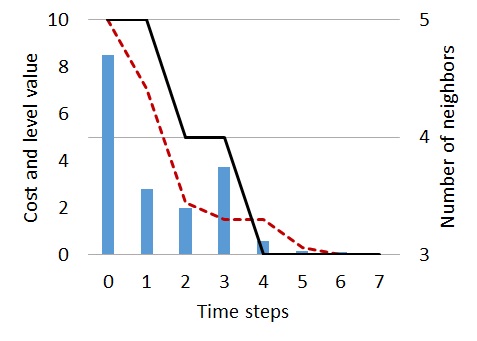}\includegraphics[width=.5\textwidth]{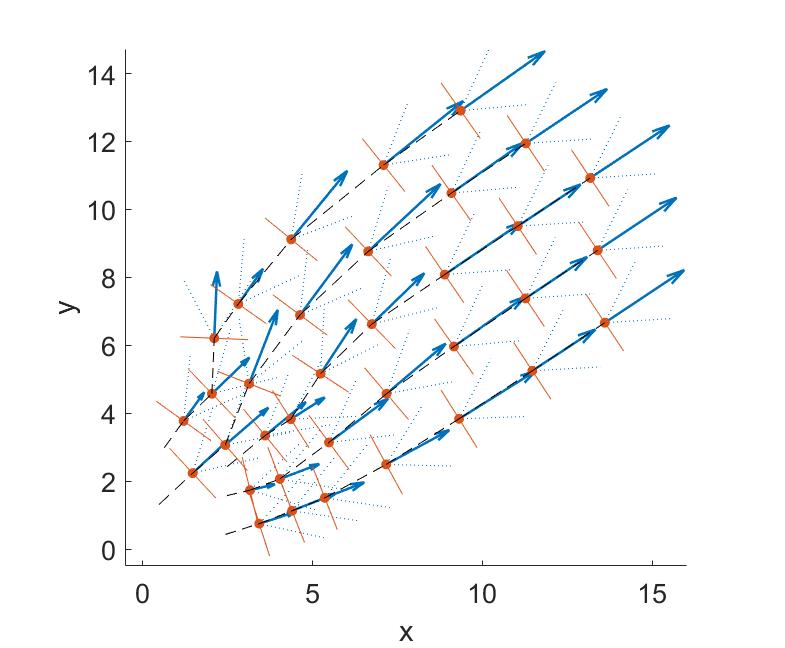}
\caption{Left: Blue bars are the values of the cost function in every time step. Red dashed line in the value of the Lyapunov function serving as a threshold for the algorithm. Black solid line is resizing of the neighborhood for the next step given the current cost. Right: Step-by-step evolution of the flock from an arbitrary initial configuration in the left lower corner towards a V-formation in the right upper corner of the plot.}
\label{fig:example}
\end{figure}

\section{Convergence and Stability}
\label{sec:convergence}

Since we are solving a nonlinear nonconvex optimization problem, the cost $J$ itself may not decrease monotonically. However, the look-ahead cost -- the cost of some future reachable state -- monotonically decreases. These costs are stored in level variables $\ell_t$ in Algorithm~\dampc and they define a Lyapunov function $V$.
\begin{eqnarray}
V(t) = \ell_t \quad \mbox{ for levels $t = 0,1,2,\ldots$} \label{func:lyap}
\end{eqnarray}
where the levels decrease by at least a minimum $\Delta$, that is, $V(t)\,{-}\,V(t\,{+}\,1)\,{>}\,\Delta$.

\begin{lemma}
$V(t): \mathbb{Z}\rightarrow\mathbb{R}$ defined by~(\ref{func:lyap}) is a valid Lyapunov function, i.e., it is positive-definite and monotonically decreases until the system reaches its goal state. 
\end{lemma}
\begin{proof}
Note that the cost function $J(\vs)$ is positive by definition, and since $l_t$ equals $J(\vs)$ for some state $\vs$, $V$ is nonnegative. Line~19 of Algorithm~\dampc\ guarantees that $V$ is monotonically decreasing by at least $\Delta$.\qed
\end{proof}


\begin{lemma}[Global Consensus]
Given Assumptions 1-7 in Section~\ref{sec:dist_flock}, all agents in the system will fix their actions in a finite number of consensus rounds.
\end{lemma}
\begin{proof}
During the first consensus round, each agent $i$ in the system runs \localampc for its own neighborhood $N_i$ of the current size $k$. Due to Lemma~\ref{lem:local_ampc}, $\exists \widehat{h}$ such that a solution, i.e. a set of action (acceleration) sequences of length $\widehat{h}$, will be found for all agents in the considered neighborhood $N_i$. Consequently, at the end of the round the solutions for at least all the agents in $N_{i^*}$, where $i^*$ is the agent which proposed the globally best solution, will be fixed. During the next rounds the procedure recurses. Hence, the set $R$ of all agents with \texttt{nfy} values is monotonically decreasing with every consensus round.\qed
\end{proof}

Global consensus is reached by the system during communication rounds. However, to achieve the global optimization goal we prove that the consensus value converges to the desired property.
\begin{definition}
Let $\{\vs(t)\,{:}\,t=1,2,\ldots\}$ be a sequence of random vector-variables and $s^*$ be a random or non-random. Then $\vs(t)$ \textbf{converges with probability one} to $s^*$ if $\mathbb{P}\left[\bigcup\limits_{\varepsilon>0}\bigcap\limits_{N<\infty}\bigcup\limits_{n\geqslant N}|\vs(t)-\vs^*|\geqslant\varepsilon\right]=0$.
\end{definition}

\begin{lemma}[Max-neighborhood convergence] If $\dampc$ is run with constant neighborhood size $B$, then it behaves identically to centralized $\ampc$.
\end{lemma}
\begin{proof}
If $\distributedampc$ uses neighborhood $B$, then it behaves like the centralized $\ampc$, because the accelerations of all birds are fixed in the first consensus round. 
\end{proof}

\begin{theorem}[Global Convergence]
Let $\M$ be an MDP $(S,A,T,J,I)$ with a positive and continuous cost function $J$ and a nonempty set of target states $G\,{\subset}\,S$, with $G\,{=}\,\{\vs\,|\,J(\vs)\,{\leqslant}\,\varphi\}$. If there exists a finite horizon $h_{max}$ and a finite number of execution steps $m$, such that centralized $\ampc$ is able to find a sequence of actions $\{\va(t):\:t=1,\ldots,m\}$ that brings $\M$ from a state in $I$ to a state in $G$, then  $\distributedampc$ is also able to do so, with probability one. 
\label{th:global_convergence}
\end{theorem}

\begin{proof}
We illustrate the proof by our example of flocking. Note that the theorem is valid in the general formulation above for the fact that as global Lyapunov function approaches zero, the local dynamical thresholds will not allow neighborhood solutions to significantly diverge from reaching the state obtained as a result of repeated consensus rounds.
Owing to Lemma~\ref{lem:local_ampc}, after the first consensus round, Alg.~\ref{alg:local_ampc} finds a sequence of best accelerations of length $h_{i^{*}}$, for birds in subflock $N_{i^{*}}$, decreasing their cost by $\Delta_{i^{*}}$. In the next consensus round, birds $j$ outside $N_{i^{*}}$ have to adjust the accelerations for their subflock $N_j$, while keeping the accelerations of the neighbors in $N_{i^{*}}\,{\cap}\,N_j$ to the already fixed solutions. If bird $j$ fails to decrease the cost of its subflock $N_j$ with at least $\Delta_j$  within prediction horizon $h_{i^{*}}$, then it can explore a longer horizon $h_j$ up to $h_{max}$. This allows \PSO to compute accelerations for the birds in $N_{i^{*}}\,{\cap}\,N_j$ in horizon interval $h_j\,{<}\,h\,{\leqslant}\,h_{i^{*}}$, decreasing the cost of $N_j$ by $\Delta_j$. Hence, the entire flock decreases its cost by $\Delta$ (this defines Lyapunov function $V$ in Eq.~\ref{func:lyap}) ensuring convergence to a global optimum. If $h_{max}$ is reached before the cost of the flock was decreased by $\Delta$, the size of the neighborhood will be increased by one, and eventually it would reach $B$. Consequently, using Theorem~1 in~\cite{attackingV}, there exists a horizon $h_{\max}$ that ensures global convergence. For this choice of $h_{max}$ and for maximum neighborhood size, the cost is guaranteed to decrease by $\Delta$, and we are bound to proceed to the next level in $\distributedampc$. The Lyapunov function on levels guarantees that we have no indefinite switching between ``decreasing neighborhood size'' and ``increasing neighborhood size'' phases, and we converge (see Fig.~\ref{fig:example_local}).\qed
\end{proof}

Fig.~\ref{fig:example_local} illustrates the proof of global convergence of our algorithm, where we overcome a local minimum by gradually adapting the neighborhood size to proceed to the next level defined by the Lyapunov function. In the plot on the right, we see $7$ birds starting from an arbitrary initial state near the origin $(x,y)=(0,0)$, and eventually reaching V-formation at position $(x,y)\approx (300,100)$. However, around $x\approx 50$, the flock starts to drift away from a V-formation, but our algorithm is able to bring it back to a V-formation. Let us see how this is reflected in terms of changing cost and neighborhood sizes. In the plot on the left, we see the cost starting very high (blue lines), but mostly decreasing with time steps initially. When we see an unexpected rise in cost value at time steps in the range $11{-}13$ (corresponding to the divergence at $x\approx 50$), our algorithm adaptively increases the horizon $h$ first, and eventually the neighborhood size, which eventually increases back to $7$, to overcome the divergence from V-formation, and maintain the Lyapunov property of the red function. Note that the neighborhood size eventually decreases to three, the minimum for maintaining a V-formation.


The result presented in~\cite{attackingV} applied to our distributed model, together with Theorem~\ref{th:global_convergence}, ensure the validity of the following corollary.
\begin{corollary}[Global Stability]
Assume the set of target states $G\in S$ has been reached and one of the following perturbations of the system dynamics has been applied: a) the best next action is chosen with probability zero (crash failure); b) an agent is displaced (sensor noise); c) an action of a player with opposing objective is performed. Then applying Algorithm~\ref{alg:distributed} the system converges with probability one from a disturbed state to a state in $G$.
\end{corollary}

\section{Experimental Results}
\label{sec:results}

We comprehensively evaluated $\distributedampc$ to compute statistical estimates of the success rate of reaching a V-formation from an arbitrary initial state in a finite number of steps $m$. 
We considered flocks of size $B=\{5,7,9\}$ birds. The specific reachability problem we addressed is as follows. 
Given a flock MDP $\M$ with $B$ birds and the randomized strategy $\sigma: S\,{\mapsto}\,\PD(A)$ of Alg.~\ref{alg:distributed}, estimate the probability of reaching a state $s$ where the cost function $J(s)\,{\leqslant}\,\varphi$, starting from an initial state in the underlying Markov chain $\M_{\sigma}$ induced by $\sigma$ on $\M$.

Since the exact solution to this stochastic reachability problem is intractable (infinite/continuous state and action spaces), we solve it approximately using statistical model checking (SMC). In particular, as the probability estimate of reaching a V-formation under our algorithm is relatively high, we can safely employ the {\em additive error} $(\varepsilon,\delta)$-Monte-Carlo-approximation scheme from~\cite{grosu2014isola}. This requires $L$ i.i.d.\ executions (up to a maximum time horizon), determining in $Z_l$ if execution $l$ reaches a V-formation, and returning the mean of the random variables $Z_1,\ldots,Z_L$. 
Recalling~\cite{grosu2014isola}, we compute $\widetilde{\mu}_Z\,{=}\,\sum_{l=1}^LZ_l/L$ by using Bernstein's inequality to fix $L{\propto}\,ln(1/\delta)/\varepsilon^2$ and obtain
$
\mathbb{P}[\mu_Z\,{-}\,\varepsilon\leq\widetilde{\mu}_Z\leq\mu_Z\,{+}\,\varepsilon]\geq{}1\,{-}\,\delta,
$
where $\widetilde{\mu}_Z$ approximates $\mu_Z$ with additive error $\varepsilon$ and probability $1\,{-}\,\delta$. 
In particular, we are interested in a Bernoulli random variable $Z$ returning 1 if the cost $J(s)$ is less than $\varphi$ and 0 otherwise. In this case, we can use the Chernoff-Hoeffding instantiation of the Bernstein's inequality~\cite{grosu2014isola}, and further fix the proportionality constant to $N\,{=}\,4\,ln(2/\delta)/\varepsilon$. 
Executing the algorithm $10^3$ times for each flock size gives us a confidence ratio $\delta\,{=}\,0.05$ and an additive error of $\varepsilon\,{=}\,10^{-2}$.

We used the following parameters: number of birds $B\,\in\,\{5,7,9\}$, cost threshold $\varphi\,{=}\,10^{-1}$, maximum horizon $h_{max}\,{=}\,3$, number of particles in \texttt{PSO} $p\,{=}\,200{\cdot}h{\cdot}B$. $\distributedampc$ is allowed to run for a maximum of $ m\,{=}\,60$ steps. The initial configurations are generated independently, uniformly at random, subject to the following constraints on the initial positions and velocities: $\forall\:i\in\{1,\ldots,B\}\:\xv_i(0)\in[0,3]\times[0,3]$ and $\vv_i(0)\in[0.25,0.75]\times[0.25,0.75]$.
To perform the SMC evaluation of $\distributedampc$, and to compare it with the centralized $\ampc$ from~\cite{attackingV}, we designed the above experiments for both algorithms in C, and ran them on the 2x Intel Xeon E5-2660 Okto-Core, 2.2 GHz, 64 GB platform. 



Our experimental results are given in Table~\ref{tab:improv}. We used three different ways of computing the average number of neighbors for successful runs. Assuming a successful run converges after $m'$ steps, we (1)~compute the average over the first $m'$ steps, reported as ``for good runs until convergence''; (2)~extend the partial $m'$-step run into a full $m$-step run and compute the average over all $m$ steps, reported as ``for good runs over $m$ steps''; or (3)~take an average across $>m$ steps, reported as ``for good runs after convergence'', to illustrate global stability.

\begin{table}[t]
	\scriptsize
    \centering
    \caption{Comparison of DAMPC and AMPC~\cite{attackingV} on $10^3$ runs.}
    \begin{tabular}{llcccccccc}
    	\toprule
        
		&& \multicolumn{3}{c}{{$\distributedampc$}} & 
		\multicolumn{3}{c}{{\textsc{$\ampc$}}} \\
		\cmidrule(l){3-5}\cmidrule(l){6-8}
		\textsc{Number of Birds}~~~~~~~ && {\centering\textsc{5}} & {\textsc{7}}& {\centering\textsc{9}}
		& {\textsc{5}} & {\textsc{7}} & {\textsc{9}} \\
		
        \midrule
        Success rate, $\widetilde{\mu}_Z$ &&
		$0.98$ & $0.92$ & $0.80$ & $0.99$ & $0.95$ & $0.88$\\
        Avg. convergence duration, $m$ &&
        $7.40$ & $10.15$ & $15.65$ & $9.01$ & $12.39$ & $17.29$\\
        Avg. horizon, $h$ &&
        $1.17$ & $1.36$ & $1.53$ & $1.27$ & $1.55$ & $1.79$\\
        Avg. execution time in sec. &&
        $295s$ & $974s$ & ${>}10^3s$ & $644s$ & $3120s$ & ${>}10^4s$\\
        
        \midrule
        Avg. neighborhood size, $k$ &&&&&&&&\\
        \midrule
        for good runs until convergence
        &&
        $3.69$ & $5.32$ & $6.35$ & $5.00$ & $7.00$ & $9.00$\\
        for good runs over $m$ steps 
        &&
        $3.35$ & $4.86$ & $5.58$ & $5.00$ & $7.00$ & $9.00$\\
        for good runs after convergence
        &&
        $4.06$ & $5.79$ & $6.75$ & $5.00$ & $7.00$ & $9.00$\\
        for bad runs &&
        $4.74$ & $6.43$ & $6.99$ & $5.00$ & $7.00$ & $9.00$\\
    
        \bottomrule
\end{tabular}
\label{tab:improv}
\end{table}

We obtain a high success rate for~5 and~7 birds, which does not drop significantly for 9 birds. The average convergence duration, horizon, and neighbors, respectively, increase monotonically when we consider more birds, as one would expect. The average neighborhood size is smaller than the number of birds, indicating that we improve over $\ampc$~\cite{attackingV} where all birds need to be considered for synthesizing the next action. 
%
We also observe that the average number of neighbors for good runs until convergence is larger than the one for bad runs, except for 5~birds. 
The reason is that in some bad runs the cost drops quickly to a small value resulting in a small neighborhood size, but gets stuck in a local minimum (e.g., the flock separates into two groups) due to the limitations imposed by fixing the parameters $h_{max}$, $p$, and $m$. The neighborhood size remains small for the rest of the run leading to a smaller average.

Finally, compared to the centralized $\ampc$~\cite{attackingV}, $\distributedampc$ is faster (e.g., two times faster for 5 birds). Our algorithm takes fewer steps to converge. The average horizon of $\distributedampc$ is smaller. The smaller horizon and neighborhood sizes, respectively, allow \texttt{PSO} to speed up its computation.

\section{Conclusions}
\label{sec:conclusion}

We introduced $\distributedampc$, a distributed adaptive-neighborhood, adaptive-horizon model-predictive control algorithm, that synthesizes actions for a controllable Markov decision process (MDP), such that the MDP eventually reaches a state with cost close to zero, provided that the MDP has such a state.

%
The main contribution of $\distributedampc$ 
is that it adaptively resizes an agent's local neighborhood, while still managing to converge to a goal state with high probability. 
Initially, when the cost value is large, the neighborhood of an agent is the entire multi-agent system. As the cost decreases, however, the neighborhood is resized to smaller values. Eventually, when the system reaches a goal state, the neighborhood size remains around a pre-defined minimal value. 

This is a remarkable result showing that the local information needed to converge is strongly related to a cost-based Lyapunov function evaluated over a global system state. While our experiments were restricted to V-formation in bird flocks, our approach applies to reachability problems for any collection of entities that seek convergence from an arbitrary initial state to a desired goal state, where a notion of distance to it can be suitably defined.

\bibliographystyle{splncs03}
\bibliography{dampc.bib}

\section*{Appendix}
\label{sec:appendix2}

\begin{figure}[t]
\centering	
\includegraphics[width=.28\textwidth]{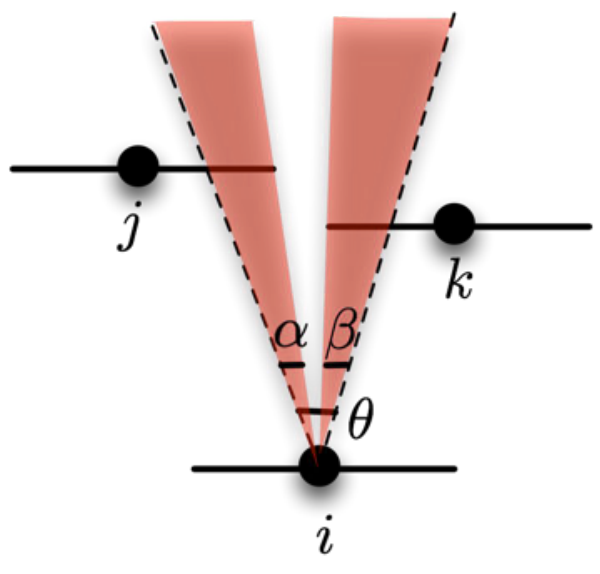}
\includegraphics[width=.3\textwidth]{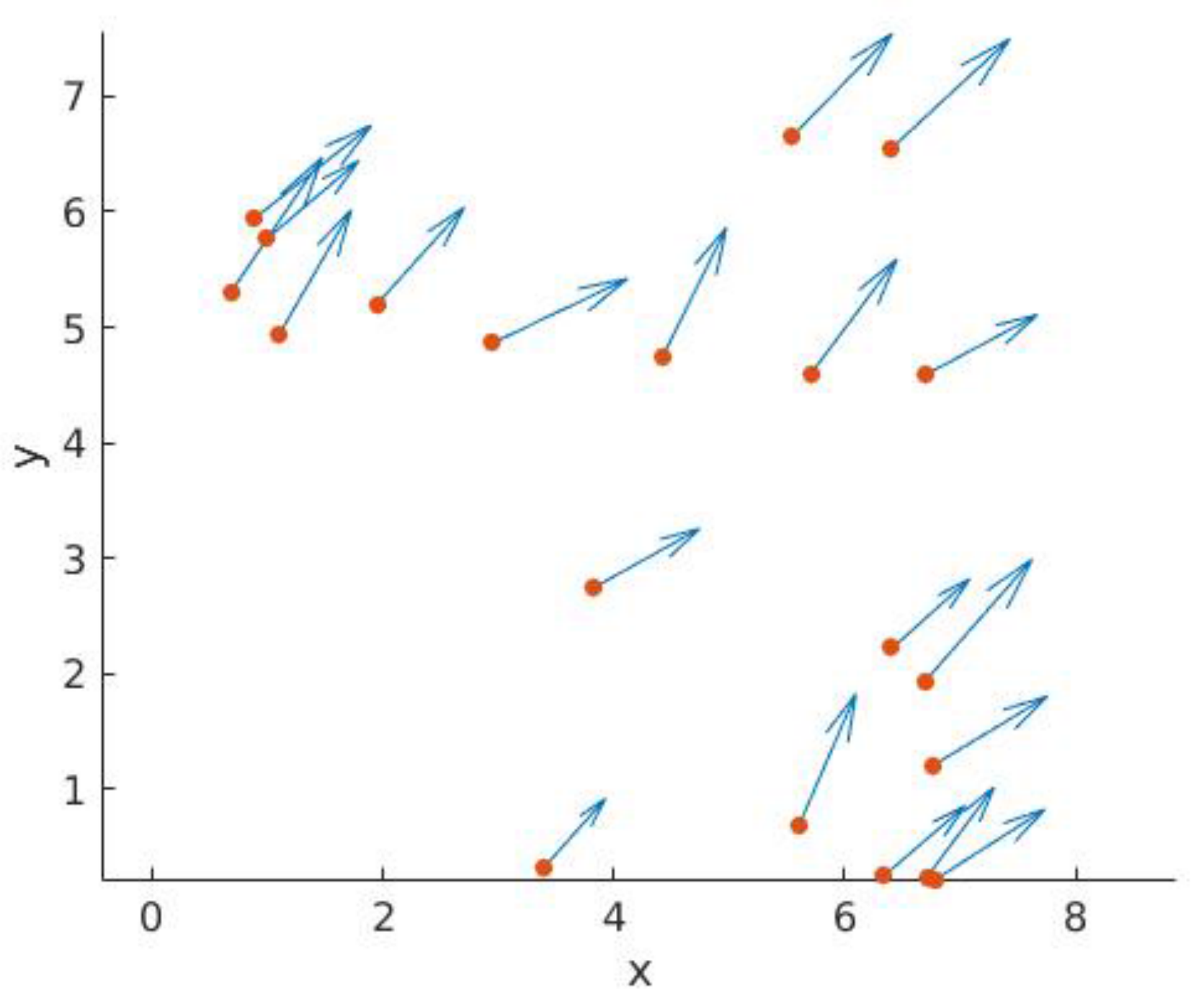}
\includegraphics[width=.36\textwidth]{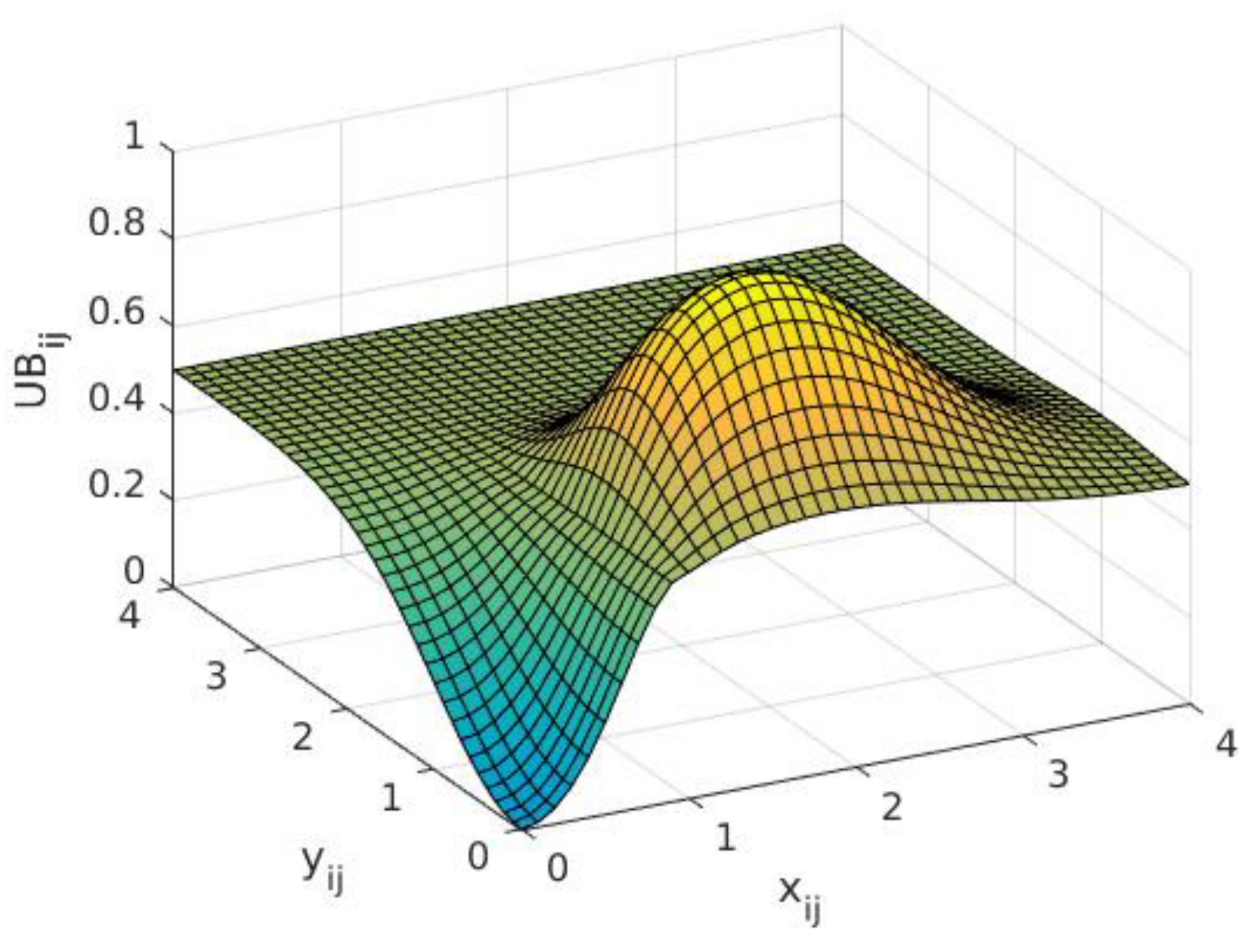}
\caption{
Illustration of the clear view ($\CV$), velocity matching  
($\VM$), and upwash benefit ($\UB$) metrics. 
Left: Bird $i$'s view is partially blocked by birds $j$ 
and $k$. Hence, its clear view is $\CV\,{=}\,(\alpha\,{+}\,\beta)/\theta$. 
Middle: A flock and its unaligned bird velocities results in a 
velocity-matching metric $\VM\,{=}\,6.2805$.  In contrast, $\VM\,{=}\,0$
when the velocities of all birds are aligned. 
Right: Illustration of the (right-wing) upwash benefit that bird $i$ receives from bird $j$ depending on how it is positioned behind bird $j$.  Note that bird $j$'s downwash region is directly behind it.
}
\label{fig:fitness}
\end{figure}
\paragraph{Clear View.} This metric is defined by accumulating the percentage of a cone with angle $\theta$, blocked by other birds: 
\begin{align*}
B_{ij}(h_{ij},v_{ij}) = \begin{cases}
   \left\{\alpha\:|\:\max\left(\frac{\pi-\theta}{2},\atan\left(\frac{v_{ij}}{h_{ij}+w}\right)\right)\leqslant\alpha\leqslant\right.\\
   \hspace*{3cm}\left.\leqslant\min\left(\frac{\pi+\theta}{2},\atan\left(\frac{v_{ij}}{h_{ij}-w}\right)\right)\right\}\\
   \hspace*{1cm}\text{if } \left(h_{ij}<w\vee\frac{h_ij-w}{v_{ij}}<\tan\theta\wedge\mathtt{Front}(j,i)\right), \\
   0 \hspace*{8mm}\text{otherwise},
  \end{cases}
\end{align*}
where $w$ represents the bird's wing span. $h_{ij}$ and $v_{ij}$ are the horizontal and vertical distances between birds $i$ and $j$, respectively, w.r.t the bird $i$'s direction. $\mathtt{Front}(j,i)=1$ when the bird $j$ is in front of the bird $i$. $\CV(\vs)$ for flock state $s$ is the sum of the clear-view metric of all birds: $\sum_i\frac{|\bigcup_{j\neq i}B_{ij}(h_{ij},v_{ij})|}{\theta}$. The minimum value is $\CV^*{=}\,0$ and attained in a perfect V-formation where all birds have an unobstructed view.

\paragraph{Velocity Matching.} $\VM(\vs)$ for flock state $\vs$ is defined as the difference between the velocity of a given bird and all other birds, summed up over all birds in the flock: $\VM(\vv)=\sum_{i>j}\left(\frac{\|v_i-v_j\|}{\|v_i\|+\|v_j\|}\right)^2$, where $v_i$ is the bird $i$'s velocity. The minimum value is $\VM^*{=}\,0$ and attained in a perfect V-formation where all birds have the same velocity.

\paragraph{Upwash Benefit.} The trailing upwash is generated near the wingtips of a bird, while downwash is in the center of a bird. An upwash measure $um$ is defined on the 2D space using a Gaussian-like model that peaks at the appropriate upwash and downwash regions: $\UB_{ij}(h_{ij},v_{ij}) =$
\begin{align*}
\begin{cases}
   \frac{v_i\cdot v_j}{\|v_i\|\cdot\|v_j\|}S(h_{ij})\cdot G(h_{ij},v_{ij},\mu_1,\Sigma_1) &\text{if } h_{ij}\geqslant\frac{(4-\pi)w}{8}\wedge\mathtt{Front}(j,i),\\
   S(h_{ij})\cdot G(h_{ij},v_{ij},\mu_2,\Sigma_2) &\text{if } h_{ij}<\frac{(4-\pi)w}{8}\wedge\mathtt{Front}(j,i),\\
   0 &\text{otherwise},
  \end{cases}
\end{align*}
\[S(h_{ij})=\mathtt{erf}\left(2\sqrt[2]{2}\left(h_{ij}-(4-\pi)w/8\right)\right),\]
\[\:G(h_{ij},v_{ij},\mu,\Sigma)=e^{-0.5([h_{ij},v_{ij}]-\mu)^T\Sigma^{-1}([h_{ij},v_{ij}]-\mu)},\]
where $h_{ij}=(4-\pi)w/8$ separates the upwash and downwash regions, $S(h_{ij})$ is a smoothing function with $\mathtt{erf}$ representing the error, and $G(h_{ij},v_{ij},\mu,\Sigma)$ is a Gaussian-like function. Means $\mu_1,\mu_2$ are chosen to maximize the upwash benefit at $[(12+\pi)w/16,1]$, and minimize it at $[0,0]$.
For bird~$i$ with upwash $um_i=\min(\sum_j\UB_{ij}(h_{ij},v_{ij},1))$, the upwash-benefit metric $\UB_i$ is $1\,{-}um_i$, and $\UB(\vs)$ for flock state $s$ is the sum of $\UB_i$ for $1\leqslant i \leqslant B$. The upwash benefit $\UB(\vs)$ in V-formation is $\UB^*\,{=}\,1$, as all birds, except for the leader, have minimum upwash-benefit metric ($\UB_i=0, um_i=1$), while the leader has an upwash-benefit metric of $1$ ($\UB_i=1, um_i=0$).

\end{document}